%% file: random-ltf.tex
\documentclass{llncs}
\usepackage{llncsdoc}

\usepackage{amssymb, amsmath, amsfonts, color, setspace}

\title{Fourier Entropy-Influence Conjecture for Random Linear Threshold Functions}

\author{Sourav Chakraborty\inst{1} \and Sushrut Karmalkar\inst{2} \and Srijita Kundu\inst{3} \and Satyanarayana V. Lokam\inst{4} \and Nitin Saurabh\thanks{The research leading to these results has received funding from the European Research Council under the European Union's Seventh Framework Programme (FP/2007-2013)/ERC Grant Agreement n. 616787.}\inst{5}}

\institute{
  Chennai Mathematical Institute, Chennai, India\\ Centrum Wiskunde Informatika, Amsterdam, Netherlands, \email{sourav@cmi.ac.in}
  \and University of Texas, Austin, USA, \email{sushrutk@cs.utexas.edu}
  \and Centre for Quantum Technologies, Singapore, \email{srijita.kundu@u.nus.edu}
  \and Microsoft Research, Bangalore, India, \email{Satya.Lokam@microsoft.com}
  \and Charles University, Prague, Czech Republic, \email{nitin@iuuk.mff.cuni.cz}
}

\newtheorem{thm}{Theorem}
\newtheorem{lem}[thm]{Lemma}
\newtheorem{coro}[thm]{Corollary}
\newtheorem{defi}[thm]{Definition}

\newcommand{\sbool}{\{+1,-1\}}

\newcommand{\calD}{\mathcal{D}}

\DeclareMathOperator{\ent}{\ensuremath \mathbb{H}}

\DeclareMathOperator{\bor}{\ensuremath \mathsf{OR}}
\DeclareMathOperator{\band}{\ensuremath \mathsf{AND}}

\newcommand{\sign}{\ensuremath {\mathsf{sign}}}

\newcommand{\boolfn}{\ensuremath {\sbool}^n \rightarrow \sbool}

\renewcommand{\inf}{\ensuremath {\mathsf{Inf}}}

\newcommand{\fhat}{\ensuremath {\hat{f}}}
\newcommand{\fc}[2]{\ensuremath {\widehat{#1}(#2)}}
\newcommand{\fcsq}[2]{\ensuremath {{\widehat{#1}}(#2)^2}}
\newcommand{\s}{\ensuremath {\mathsf{s}}}
\newcommand{\as}{\ensuremath {\mathsf{as}}}

\newcommand{\E}{\ensuremath {\mathbb E}}

\renewcommand{\leq}{\leqslant}
\renewcommand{\geq}{\geqslant}

\newcommand{\skipthis}[1]{}

\newcommand{\reals}{\ensuremath \mathbb{R}}

\begin{document}
\maketitle

\input{abstract}

\input{intro}

\input{lb-inf}

\input{lb-random}
\input{conclusion}

\paragraph{Acknowledgments.} We thank the reviewers for helpful comments that improved the presentation of the paper.

\bibliographystyle{splncs}
\bibliography{entropy}

\appendix
\section{Reducing Non-Homogeneous to Homogeneous case}
\label{subsec:unbal}

Let $f(x) = \sign(w_{0} + \sum_{i=1}^n w_ix_i)$ for all $x \in \sbool^n$.
Recall from Eq.~\eqref{eq:inf-ltf} we have an exact expression for the
$i$-th influence for all $1 \leq i \leq n$. We can relax the probability estimate to lower bound the influence as follows, 
\begin{align}
  \label{eq:non-hom-inf}
\inf_i(f)  
  \geq \frac{1}{2}\Pr_{x \in \sbool^n}\left[\left|w_0 + \sum_{j  \leq n \colon j \neq i} w_jx_j\right|\leq  |w_i| \right] \enspace . 
\end{align}
Now consider the function $g(x_0, x_1, \dots, x_n) = \sign(\sum_{i=0}^{n} w_ix_i)$ by adding the extra variable $x_{0}$. We claim that $\inf_i(g)\leq 2 \inf_i(f)$, for all $1\leq i \leq n$. Fix an $ i \in [n]$. 
From Eq.~\eqref{eq:inf-ltf} we know that $\inf_i(g)$ equals 
\[\frac{1}{2}\left(\Pr\left[- |w_i| < w_0 + \sum_{1 \leq j  \leq n \colon j \neq i} w_jx_j \leq  |w_i| \right] + \Pr\left[-|w_i| < - w_0 +\sum_{1 \leq j  \leq n \colon j \neq i} w_jx_j \leq  |w_i| \right]\right),
\]
where the probabilities are  uniform distribution over
$x_1, \dots, x_n \in \sbool^n$. 
By relaxing the event in each case 
we have
\begin{align*}
  \inf_i(g) & \leq \frac{1}{2}\left(\Pr\left[- |w_i| \leq w_0 + \sum_{1 \leq j  \leq n \colon j \neq i} w_jx_j \leq  |w_i| \right] + \Pr\left[-|w_i| \leq - w_0 +\sum_{1 \leq j  \leq n \colon j \neq i} w_jx_j \leq  |w_i| \right]\right).
\end{align*}
It is easily seen that, 
\[\Pr_{x \in \sbool^{n}}\left[-|w_i| \leq w_0+\sum_{1\leq j  \leq n \colon j \neq i} w_jx_j \leq  |w_i| \right] = \Pr_{x \in \sbool^{n}}\left[-|w_i| \leq -w_0 + \sum_{1\leq j  \leq n \colon j \neq i} w_jx_j \leq  |w_i| \right] \enspace.\]
Indeed, there exists a 1-1 correspondence between $n$-bit strings
satisfying the left hand side event and the right hand side event. 
Thus,
\[\inf_i(g) \leq \Pr_{x \in \sbool^{n}}\left[-|w_i| \leq w_0+\sum_{1\leq j  \leq n \colon j \neq i} w_jx_j \leq  |w_i| \right] \leq 2\cdot  \inf_i(f) \enspace,\]
where the second inequality follows from~\eqref{eq:non-hom-inf}. 

Therefore, we have \[\inf(f) \geq \sum_{i=1}^{n}\frac{\inf_i(g)}{2} = \frac{\inf(g) - \inf_0(g)}{2} \geq \frac{\inf(g) -1}{2} \enspace .\]
Thus we can translate a lower bound on influences in the homogeneous case to the non-homogeneous case. 
\end{document}

%% file: abstract.tex
\begin{abstract}
  The Fourier-Entropy Influence (FEI) Conjecture states that for any
  Boolean function $f:\sbool^n \to \sbool$,
  the Fourier entropy of $f$ is at most its influence
  up to a universal constant factor.
  While the FEI conjecture has been proved 
  for many classes of Boolean functions, it is still not known whether it
  holds for the class of Linear Threshold Functions.
  A natural question is:
  Does the FEI conjecture hold for a ``random'' linear threshold function?
  In this paper, we answer this question in the affirmative.
  We consider two natural distributions on the weights defining a linear
  threshold function, namely uniform distribution on $[-1,1]$
  and Normal distribution.  
\end{abstract}

%% file: intro.tex
\section{Introduction}
\label{sec:intro}

Boolean functions are the most basic object of study in Theoretical
Computer Science. There are many complexity measures associated with
a Boolean function  $f \colon \boolfn $, e.g., \emph{degree},
\emph{sensitivity}, \emph{certificate complexity}, etc. These measures
provide lower bound on the complexity of a Boolean function in different models
of computation (e.g, PRAM, Circuits/Formulas, Decision trees, etc.).

A particularly interesting one is the
\emph{average sensitivity} of
a Boolean function. Given $f \colon \boolfn $, and an input $x \in \sbool^n$,
sensitivity of $f$ at $x$, denoted $\s(f,x)$,
is defined as the number of neighbours of $x$
in the Hamming cube where $f$ takes different value than at $x$.
The \emph{average sensitivity} of $f$, denoted $\as(f)$,
is the average of $\s(f,x)$, $\E_x[\s(f,x)]$, under the uniform distribution.
In the definition above, we fixed an $x \in \sbool^n$ and looked at
the variables the value of $f(x)$ depends on. Similarly, we can
fix a variable $i \in [n]$, and look at the number of inputs $x$
such that the value $f(x)$ depends on this variable. This
leads us to the notion of \emph{influence}. The \emph{influence}
of the $i$-th variable, denoted $\inf_i(f)$, is defined to be
$\Pr_x[f(x) \neq f(x^i)]$, where $x$ is chosen uniformly at random and
$x^i$ denotes $x$ with $i$-th variable negated. The \emph{influence}
of $f$, $\inf(f)$, is defined as $\sum_{i=1}^n \inf_i(f)$. It is easily seen that
$\inf(f)$ equals $\as(f)$. Influence is related to many complexity measures,
and many of these relations can be established via Fourier analysis.
For example, circuit size~\cite{LMN93,Boppana97}, formula size~\cite{GKLR12},
decision trees~\cite{OSS05}, etc. For an in-depth treatment of
Fourier analysis in Boolean functions we refer to \cite{ODonnell-book}.

Every Boolean function $f \colon \boolfn $ has a unique
representation in the Fourier basis $\colon$
\[f(x) = \sum_{S \subseteq [n]}\fc{f}{S}\prod_{i \in S}x_i \enspace .\]
The set of Fourier coefficients is given by $\{\fc{f}{S}\}_{S \subseteq [n]}$.
It follows from Parseval's identity that
$\sum_{S\subseteq [n]}\fcsq{f}{S} = \E_x[f(x)^2] = 1$. Thus, the squared Fourier
coefficients naturally define a distribution over the subsets of
$\{1,2,\ldots ,n\}$. The Fourier-entropy of $f$, denoted $\ent(f)$, is
defined to be the Shannon entropy of this distribution. That is,
\[\ent(f) := \sum_{S\subseteq [n]} \fcsq{f}{S} \log \frac{1}{\fcsq{f}{S}} \enspace.\]
A longstanding and important conjecture in Analysis of Boolean functions
states that the Fourier entropy of a Boolean function is bounded above by
the total influence of the function up to a constant factor.
More formally,
\paragraph{Fourier-Entropy Influence (FEI) Conjecture}:
There exists a universal constant $C > 0$ such that for all
$f \colon \boolfn$,~ $\ent(f)  \leq C \cdot \inf(f)$. That is,
\begin{align}
\label{eq:FEI-conj}
\sum_{S\subseteq [n]} \fcsq{f}{S} \log \frac{1}{\fcsq{f}{S}} \leq C \cdot
\sum_{S\subseteq [n]}|S|\fcsq{f}{S} \enspace .
\end{align}

The conjecture was made by Friedgut and Kalai~\cite{FK96} in 1996.
The genesis of the conjecture is in the study of threshold phenomena in
random graphs~\cite{FK96}.
For example, it implies a lower bound of $\Omega(\log^2 n)$ on the
influence of any $n$-vertex monotone graph property.
The current best lower bound, by Bourgain and Kalai~\cite{BK97}, is
$\Omega(\log^{2-\epsilon} n)$, for any constant $\epsilon > 0$.
However, over time, many non-trivial implications have been observed.
In particular, it implies a variant of
\emph{Mansour's Conjecture}~\cite{Mansour95} that is sufficient
to imply a polynomial time \emph{agnostic} learning algorithm
for DNFs~\cite{GKK08}, for any constant error parameter,
resolving a major open problem~\cite{GKK08-colt}
in computational learning theory.
\paragraph{Mansour's Conjecture (variant)}:
Let $f \colon \boolfn$ be any Boolean
function computable by a $t$-term DNF formula. Then, for any $\epsilon>0$,
there exists a polynomial
$p$ with $t^{O(1/\epsilon)}$ terms such that $\E[(f-p)^2] \leq \epsilon$.
(Mansour conjectured the exponent of $t$ to be $O(\log \frac{1}{\epsilon})$
\cite{Mansour95}.)

In general, the FEI conjecture implies
\emph{sparse} $L_2$-\emph{approximations} for Boolean functions.
That is, it implies the existence of a polynomial $p$
with $2^{O(\inf(f)/\epsilon)}$ terms such that $\E[(f-p)^2] \leq \epsilon$,
for any $\epsilon > 0$. Presently, the best known
construction~\cite{Friedgut98} yields a
bound of $2^{O((\inf(f)/\epsilon)^2)}$ on the number of terms.

Finally, the FEI inequality~\eqref{eq:FEI-conj} is also known to imply
the famous Kahn-Kalai-Linial theorem~\cite{KKL88}.

\emph{KKL Theorem }: For any Boolean function $f \colon \boolfn$,
there exists an $i \in [n]$
such that \(\inf_i(f) = \Omega\left(\mathrm{Var}(f)\frac{\log n}{n}\right)\).

In fact,
as observed by~\cite{OWZ11}, the following weakening of
the FEI conjecture suffices to imply the KKL theorem.
\paragraph{Fourier Min-Entropy-Influence (FMEI) Conjecture}: There exists a universal
constant $C > 0$ such that for all $f \colon \boolfn$,
\[\min_{S \subseteq [n]}\log \frac{1}{\fcsq{f}{S}}  \leq C \cdot \inf(f) \enspace .\]
In other words, it conjectures how large the maximum Fourier coefficient
must be. That is, there exists a set $S \subseteq [n]$ such that
$\fcsq{f}{S} \geq 2^{-C\cdot \inf(f)}$.

For more on the FEI conjecture we refer the interested reader to Gil Kalai's
blog post~\cite{Kalai-blog}.

While the FEI and the FMEI conjectures have resisted solutions for long,
it is easy to verify the conjectures for simple functions such as
$\bor$, $\band$, $\mathsf{Tribes}$, $\mathsf{Majority}$, etc.
Thus, researchers have tried to establish the conjectures for special
classes of Boolean functions. In particular, it has been shown that FEI conjecture holds
for random DNFs~\cite{KLW10}, symmetric functions and read-once decision
trees~\cite{OWZ11}, random functions~\cite{DPV11},
read-once formulas~\cite{OT13,CKLS16}, decision trees with bounded average depth~\cite{CKLS16,www14},
and bounded read decision trees~\cite{www14}.  In \cite{OWZ11} it was also
observed that FMEI conjecture holds for monotone functions.

In this paper, we study the FEI conjecture for the class of
\emph{linear threshold functions} (LTF). A Boolean function $f \colon \boolfn$
is said to be an LTF if there exists $w_0,w_1,\ldots ,w_n \in \reals$ such
that for all $x \in \sbool^n$, $f(x) = \sign(w_0 + w_1x_1+ \cdots + w_nx_n)$, where, the $\sign$ is the
function that maps positive values to $+1$ and negative values to $-1$. We assume, without loss of generality,
$\sign(0) = -1$.

We observe that the KKL theorem implies Fourier Min-Entropy-Influence conjecture holds for
linear threshold functions using the same idea from \cite{OWZ11}. But for the case of the FEI conjecture we still
cannot show that it holds for linear threshold functions in general. So a natural question to ask is:\smallskip 

\begin{center}
Does FEI conjecture hold for a random LTF?
\end{center}

Similar question has been asked and answered for other classes of functions (for random DNFs~\cite{KLW10} and for random Boolean functions~\cite{DPV11}).
A natural way of defining a random LTF is when the coefficients $w_0, \dots, w_n$ are drawn from a distribution $\calD$.
Choosing $w_0,\ldots ,w_n$ is equivalent to choosing a direction in
$(n+1)$-dimensions. Thus, a natural way of sampling an LTF is to sample
a unit vector in $(n+1)$-dimensions. There are many ways to sample a unit vector uniformly~\cite{HW59,Muller59,Mars72,PB97}.
It is well known that a simple way to sample a unit
vector uniformly in $(n+1)$-dimensions is to sample each coordinate
from standard normal distribution and then normalize the vector.
Specifically, we consider for $\calD$ the standard normal distribution and
the uniform distribution over $[-1,1]$. To establish our main result,
Theorem~\ref{thm:main}, we prove a generic technical result
that says that the FEI conjecture holds with very high probability
as long as $\calD$ possesses some ``nice'' properties. Informally, these properties say that
a centered random variable with variance 1 and bounded third absolute moment has at least a constant probability mass above the third absolute moment. Our main results are stated below for two specific natural distributions.

\begin{thm}[Main]
  \label{thm:main}[Informal]
  \begin{enumerate}
  \item   If $w_0, \dots, w_n$ are drawn from the normal distribution  $N(0,1)$  and $f(x) = \sign(w_0 + w_1x_1+ \cdots + w_nx_n)$ then
    with high probability the FEI Conjecture holds for $f$.
  \item   If $w_0, \dots, w_n$ are drawn from the uniform distribution  over $[-1,1]$  and $f(x) = \sign(w_0 + w_1x_1+ \cdots + w_nx_n)$ then
    with high probability the FEI Conjecture holds for $f$.
  \end{enumerate}
\end{thm}

We also identify (Corollary~\ref{cor:tau-regular}) certain subclasses of linear threshold functions for which FEI conjecture holds.

\subsection{Our Proof Technique}

To prove the FEI conjecture for a function $f$ (or, any class of Boolean functions) one needs to provide an upper bound on the Fourier entropy of $f$ and a
matching (up to a constant factor) lower bound on the influence of $f$.
For upper bounding the Fourier entropy of a LTF we crucially use
the following theorem proved in \cite{CKLS16}.
\begin{thm}\label{thm:entub}\cite{CKLS16} If $f:\sbool^n \to \sbool$ is \emph{any} linear threshold function, then $\mathbb{H}(f) \leq C\cdot \sqrt{n},$ where, $C$ is a universal constant.
\end{thm}
Thus if $f$ is an LTF with influence lower bounded by $\Omega(\sqrt{n})$, then using Theorem~\ref{thm:entub}, we conclude that FEI conjecture holds for $f$.
Our main contribution in the paper is to prove that a ``random'' LTF has influence $\Omega(\sqrt{n})$.

We provide two techniques of proving a lower bound on the influence of an LTF $f$. The first technique (Theorem~\ref{thm:InfKhint}) is a simple application of the Khintchine Inequality \cite{Sza76}.
%
While this lower bound technique is easy and simple, this does not yield
a high probability statement about the event
$\inf(f) = \Omega(\sqrt{n})$
when the coefficients are distributed according to the normal distribution
(see Section~\ref{subsec:lb-normal}).
For this, we need the second technique. 

The starting point in this case is the following basic inequality for LTF's.
\begin{align}
  \label{eq:inf-ltf}
  \inf_i(f) = \Pr_{x \in \sbool^n}\left[-|w_i| < w_0 + \sum_{j  \leq n \colon j \neq i} w_jx_j \leq |w_i| \right] \enspace.
\end{align}
We bound this probability using concentration inequalities.  In particular, we crucially use an
optimal version of the Berry-Esseen Theorem (Theorem~\ref{thm:shevtsova}) that was proved  recently by Shevtsova~\cite{shevtsova12}. Using the Berry-Esseen Theorem we show, informally,
if a variable has a ``high'' weight then it has ``high'' influence.
More precisely, we prove the following technical lemma en route to
establishing the main theorem. Consider a symmetric
distribution $\calD$ around 0 with variance 1.
\begin{lem}
  Let $w_j \sim \calD$ for $1 \leq j \leq n$. For all $i \in [n]$,
  $\alpha\in \reals^{+}$, and $\delta > 0$,  with probability at
  least $1 - e^{-\Omega(n \mu_3^2)}$ 
  over the choices of $w_j$'s,
  \[ \Pr_{x \in \sbool^n}\left[\left|\sum_{1\leq j \leq n \colon j \neq i} w_jx_j\right| \leq  \alpha\right] \geq \frac{\theta}{\sqrt{n}} - O\left(\frac{1}{n^{2/3}}\right),\]
  where $\theta = (\alpha - \mu_3(1+ \frac{2\delta}{1-\delta}))/\sqrt{2\pi (1+\delta)}$, and $\mu_3 = \E_{w \sim \calD}[|w|^3]$.
\end{lem}

Observe that for the statement to be non-trivial $\alpha > \mu_3+\delta'$ for some $\delta' > 0$. The optimality of Shevtsova's theorem (Theorem~\ref{thm:shevtsova}) is crucial to the fact that $\alpha$ can be chosen to be
only slightly larger than $\mu_3$, by an \emph{additive} constant rather than 
$C \cdot \mu_3$ for some large constant $C > 1$. This difference is
significant when $\calD$ is a \emph{truncated} distribution.
A truncated distribution is a conditional distribution that results from
restricting the support of some other probability distribution. 

\paragraph{Organization of the paper:}
In Section~\ref{sec:lb-inf} we present the two techniques to 
lower bound 
the influence of an LTF. We then show that
the FEI conjecture holds for a `random' LTF in Section~\ref{sec:lb-inf-random},
and conclude with some observations in Section~\ref{sec:conclusion}.

%
%


%% file: lb-inf.tex
\section{Lower Bounds on Influence}
\label{sec:lb-inf}
Recall that $f \colon \boolfn$ is a linear threshold function if there
exists $n+1$ real numbers, $w_0, w_1,\ldots , w_n$, such that
\[\forall x=x_1\dots x_n \in \sbool^n \colon  f(x) = \sign (w_0 + w_1x_1 + \cdots + w_nx_n) \enspace ,\]
where $\sign(z) = 1$ if $z > 0$, and $-1$ if $z \leq 0$.
%
It is clear that there are infinitely many $(n+1)$-tuples that define the
same Boolean function $f$.
Nevertheless, relationships among $w_i$'s characterize
many properties of the function. In fact, the influence of $f$ is a function
on the $w_i$'s. 
Our goal is to obtain lower bounds for this function using various properties of $w_i$'s. 
In this section we present two such lower bounds. 

The first lower bound that we obtain is a simple observation and the lower bound is in terms of the $\ell_2$-norm of the $w_i$'s and
the maximum value of the $w_i$'s. 
\begin{thm}
  \label{thm:InfKhint}
  Let $f \colon \boolfn$ be such that 
  $f(x) = \sign (w_0 + \sum_{i=1}^n x_iw_i)$ for some weights 
  $w_0, w_1,\ldots , w_n \in \reals$. 
  Then, 
  \[\inf(f) \geq \frac{\sqrt{\sum_{i=0}^n w_i^2}}{2\sqrt{2}\max_{i} |w_i|} - |\fhat(\emptyset)| \geq \frac{\sqrt{\sum_{i=0}^n w_i^2}}{2\sqrt{2}\max_{i} |w_i|} - 1 \enspace .\]
\end{thm}
The proof is omitted. 
It 
follows from an easy application of the well-known Khintchine Inequality. 
\begin{thm}[Khintchine Inequality]\cite{Sza76}
  \label{thm:khintchine}
  Let $w_1,w_2,\ldots , w_n \in \reals$ be such that $\sum_{i} w_i^2 = 1$. Then,~
  \(\E_{x \in \sbool^n} [|w_1x_1+\cdots +w_nx_n|] \geq \frac{1}{\sqrt{2}}.\)
\end{thm}

The second lower bound we obtain is more sophisticated lower bound. 
Recall that the $i$-th influence of $f$ is given by
\( \Pr_{x \in \sbool^n}[ f(x) \neq f(x^i)]\).  Using the fact that $f$ is an LTF,
it is easily seen to be equivalent to Eq.~\eqref{eq:inf-ltf}. 
%
We bound the 
expression in~\eqref{eq:inf-ltf} from below to establish
the lower bound. 
\begin{thm}
  \label{thm:mainlb}
Let $w_0, w_1,\ldots , w_n \in \reals$  and let $f:\sbool^n \to \sbool$ be defined by $f(x) = \sign(w_0 + \sum_{i=1}^n x_iw_i)$. Then, for $1\leq i \leq n$,  
\[ \inf_i(f) \geq \frac{1}{2\sqrt{2\pi (n-1)B_i}}\left(|w_i| - \frac{A_i}{B_i}\right) - \frac{|w_i|^3}{6\sqrt{2\pi}((n-1)B_i)^{3/2}}  - \frac{3.4106~A_i^{4/3}}{(n-1)^{2/3}B_i^2},\]
where $A_i := \frac{1}{n-1}\sum_{0 \leq j\neq i \leq n}|w_j|^3$
and $B_i := \frac{1}{n-1}\sum_{0\leq j\neq i \leq n} |w_j|^2$.
\end{thm}
A crucial ingredient in the proof is the following optimal version of Berry-Esseen
Theorem 
that bounds the uniform distance
between the cumulative distribution function of
the standard normal distribution
$N(0,1)$ and
the standardized sum of independent symmetric Bernoulli random variables.
\begin{thm}[Corollary~4.19 in~\cite{shevtsova12}]
  \label{thm:shevtsova}
  Let $X_1,\ldots , X_n$ be independent symmetric Bernoulli random variables
  such that for all $1\leq j \leq n$ there exist $a_j$ with
$\Pr[X_j = a_j] = \Pr[X_j=-a_j] = \frac{1}{2}$.
If $\Phi(x)$ is the cumulative distribution function of the standard normal distribution, then,
  \[\sup_x\left|\Pr\left[\frac{X_1+\cdots +X_n}{\sqrt{\sum_{j=1}^n a_j^2}} < x\right] - \Phi(x)\right| \leq \frac{\ell_n}{\sqrt{2\pi}} + 3.4106 \cdot \ell_n^{4/3} \enspace ,\] 
where $\ell_n = (\sum_j |a_j|^3)/(\sum_j a_j^2)^{3/2}$.
\end{thm}
The optimality is in the fact that the constant of $1/\sqrt{2\pi}$ (as a coefficient of $\ell_n$) is the best possible.

We will prove Theorem~\ref{thm:mainlb} 
in two stages. First we establish it 
under the assumption that $w_0 = 0$.  
Then, in Section~\ref{subsec:unbal}, we reduce
the non-homogeneous case ($w_0 \neq 0$) to the homogeneous case ($w_0 = 0$)
to complete the proof. 

%
%

\subsection{Proof of Theorem~\ref{thm:mainlb} (assuming $w_0 = 0$)}
\label{subsec:balanced-fixed-weights}

As mentioned earlier, the full proof of Theorem~\ref{thm:mainlb} will follow
from the discussions in Section~\ref{subsec:unbal}.
Thus, for some $w_1,\ldots, w_n \in \reals$,
we have $f(x) = \sign(\sum_{i=1}^n w_ix_i)$ for all $x \in \sbool^n$.
Relaxing the estimate in Eq.~\eqref{eq:inf-ltf}, we have 
\[\inf_i(f)
\geq \frac{1}{2}\Pr_{x \in \sbool^n}\left[\left|\sum_{j  \leq n \colon j \neq i} w_jx_j\right| \leq  |w_i| \right] \enspace .\]
To obtain the claimed lower bound on influence, we further lower bound the
above expression. 
We start with some definitions required
for the proof. 

For any $1\leq i\leq n$, we define
\[A_i:=\frac{1}{n-1}\sum_{j\neq i}|w_j|^3,\;\; \mbox{ and } \;\;\; B_i := \frac{1}{n-1}\sum_{j\neq i} |w_j|^2.\]

Note that the proof of Theorem~\ref{thm:mainlb}, for the case
when $w_0 = 0$, follows by substituting $\alpha = |w_i|$ in the following lemma.

\begin{lem}
  \label{lem:mainlb}
  For all $i \in [n]$ and any $\alpha \in \mathbb{R}^+$,
  \[ \frac{1}{2}\Pr_{x \in \sbool^n}\left[\left|\sum_{j \leq n \colon j \neq i} w_jx_j\right| \leq \alpha \right] \geq \frac{1}{\sqrt{2\pi (n-1)B_i}}\left(\alpha - \frac{A_i}{B_i}\right)  - \frac{\alpha^3}{6\sqrt{2\pi}((n-1)B_i)^{3/2}}  - \frac{3.4106~A_i^{4/3}}{(n-1)^{2/3}B_i^2} \enspace .\]
\end{lem}
\begin{proof}
  For any $\alpha>0$, by the symmetry of distribution over $\sum w_jx_j$,
  we have
  \begin{align*}
    \frac{1}{2}\Pr_x\left[\left|\sum_{j \leq n \colon j \neq i} w_jx_j \right| \leq \alpha \right] & = \Pr_x\left[0 \leq \sum_{j \leq n \colon j \neq i} w_jx_j \leq \alpha \right]  \\
    & = \Pr_x\left[\frac{\sum_{j \leq n \colon j \neq i} w_jx_j}{\sqrt{\sum_{j \leq n \colon j \neq i} w_j^2}}\leq \frac{\alpha}{\sqrt{\sum_{j \leq n \colon j \neq i} w_j^2}} \right] - \frac{1}{2} \enspace . 
  \end{align*}
  Using Theorem~\ref{thm:shevtsova} on random variables $w_jx_j$'s,
  we obtain the following lower bound
  \begin{align}
    \label{eq:main-eq}
    \frac{1}{2}\Pr_x\left[\left|\sum_{j \leq n \colon j \neq i} w_jx_j \right| \leq \alpha \right] & \geq \Phi \left(\frac{\alpha}{\sqrt{\sum_{j \leq n \colon j \neq i} w_j^2}}\right) -\left(\frac{1}{2}\right) - \frac{L_{(n-1)}}{\sqrt{2\pi}} - 3.4106 \cdot L_{(n-1)}^{4/3} \enspace, 
  \end{align}
  where
  \[ L_{(n-1)} = \frac{\sum_{j \leq n \colon j \neq i} |w_j|^3}{(\sum_{j \leq n \colon j \neq i} |w_j|^2)^{3/2}}=\frac{(n-1)A_i}{((n-1)B_i)^{3/2}} \enspace .\]
Thus we have 
  \begin{align}
    \label{eq:part2}
    - \frac{L_{(n-1)}}{\sqrt{2\pi}} - 3.4106 \cdot L_{(n-1)}^{4/3} & = -\frac{1}{\sqrt{2\pi}}\cdot  \frac{A_i}{\sqrt{(n-1)}B_i^{3/2}} - 3.4106\frac{A_i^{4/3}}{(n-1)^{2/3}B_i^2} \enspace .  
  \end{align}
  Also, using the following Maclaurin series
  \[\Phi(x) - \frac{1}{2} = \frac{1}{\sqrt{2\pi}} \left(x - \frac{1}{6}x^3 + \frac{1}{40}x^5 - \cdots\right) \enspace ,\]
  we have
  \begin{align}
    \label{eq:part1}
    \Phi \left(\frac{\alpha}{\sqrt{\sum_{j \leq n \colon j \neq i} w_j^2}}\right) -\frac{1}{2} \geq  \frac{\alpha}{\sqrt{2\pi B_i(n-1)}} - \frac{\alpha^3}{6\sqrt{2\pi}((n-1)B_i)^{3/2}} \enspace .
  \end{align}

Using equations~\eqref{eq:part1} and \eqref{eq:part2} in
  Equation~\eqref{eq:main-eq}, we obtain 
  \begin{align*}
    \frac{1}{2}\Pr_x\left[\left|\sum_{j \leq n \colon j \neq i} w_jx_j \right| \leq \alpha \right] & \geq \frac{1}{\sqrt{2\pi(n-1) B_i}}\left( \alpha-\frac{A_i}{B_i}\right) - \frac{\alpha^3}{6\sqrt{2\pi}((n-1)B_i)^{3/2}}  - \frac{3.4106~A_i^{4/3}}{(n-1)^{2/3}B_i^2} \enspace . 
  \end{align*}
\qed\end{proof}

In our applications of the lemma, the second and third terms will be negligible
compared to the first term. 

\subsection{Fourier Entropy Influence Conjecture for a class of Linear Threshold Functions}
In this section we identify a class of LTFs for which we prove
the FEI conjecture. The class of functions is $\tau$-regular functions. 
\begin{defi}
  Suppose $w_0, w_1, \dots, w_n$ are real numbers such that $\sum_{i=0}^n w_i^2 = 1$. If for all $0\leq i\leq n$, $|w_i|\leq \tau$ then the linear threshold function, $f$ defined as
  $f(x_1, \dots, x_n) = \sign(w_0 + \sum_{i=1}^n w_ix_i)$ is called a $\tau$-regular LTF.
\end{defi}
From Theorem~\ref{thm:InfKhint} and Theorem~\ref{thm:entub} we can show that $\tau$-regular LTFs satisfy the FEI conjecture. 
\begin{coro}
  \label{cor:tau-regular}
  If $\tau \leq c/\sqrt{n}$ and if $f$ is a $\tau$-regular function then 
  $$\mathbb{H}(f)\leq O(\inf(f)),$$ where the constant in the Big-oh notation depends on $c$. 
\end{coro}


%
%
%
%

%% file: lb-random.tex
\section{Lower Bounds on Influence for Random Linear Threshold Functions}
\label{sec:lb-inf-random}

Given that we still cannot prove the FEI conjecture for all Linear Threshold
Functions a natural question is whether FEI conjecture holds
for a ``random'' LTF. 

Suppose $w_0, w_1, \dots, w_n$ are drawn independently from a distribution
$\calD$. 
Consider the function $f(x)$ where 
\[\forall x=x_1, \dots, x_n \in \sbool^n \colon  f(x) = \sign (w_0 + w_1x_1 + \cdots + w_nx_n).\]
Without loss of generality we can assume that the distribution is symmetric around the origin and thus 
the mean of the distribution $\mu(\calD)$ is $0$.
Also since the function remains same even if 
we scale the $w_i$'s by any value, so we can also assume that the variance of the distribution $\sigma(\calD)$ is $1$. 

As a step towards proving the FEI conjecture for $f$ we need to lower bound the influence of $f$. 
Note that Theorem~\ref{thm:InfKhint} 
and Theorem~\ref{thm:mainlb} give lower bounds on the 
influence of $f$ in terms of $w_i$'s. Using similar arguments we can obtain a lower bound on the influence that 
holds for a random $f$ (that is, when $w_i$'s are drawn from a distribution $\calD$) with high probability. 

For any $n \in \mathbb{N}$,
let $w_1,\ldots , w_n  \sim \calD$ be the outcome
of $n$ independent samples according to $\calD$. 
 We define, for any $\alpha \in \reals^{+}$,   
$p_{\calD,n}(\alpha)=\Pr[\max_{i=1}^n\{|w_i|\} \geq \alpha]$. 

\begin{coro}[Corollary of Theorem~\ref{thm:InfKhint}]\label{cor:random-Khint}
For any $\alpha\in \reals^{+}$, with probability at least $1-p_{\calD,n+1}(\alpha) - o(1)$, 
\[\inf(f) \geq \Omega(\sqrt{n}/\alpha).\]
\end{coro}
Using Bernstein's inequality it is easily seen that
$\sum_iw_i^2 = \Omega(n)$ with probability $1 - o(1)$. 
\begin{thm}\label{thm:lb_random}
  For any $\alpha \in \reals^{+}$ and any $\delta > 0$,
  with probability at least
  $1 - e^{-\delta^2 n \mu_3^2/\sigma_3} - 2e^{-\delta^2 n/\sigma_2} - 2e^{-(n/4) p_{\calD,1}(\alpha)^2}$ over the choices of $w_j$'s 
\[\inf(f) \geq p_{\calD,1}(\alpha) \cdot \Theta(\alpha) \cdot \sqrt{n},\]
where 
$\Theta(\alpha) =  \left(\alpha - \mu_3(1+\frac{2\delta}{1-\delta})\right)/\sqrt{2\pi (1 + \delta)}$, $\mu_3 = \E_{w\sim \calD}[|w|^3]$, and $\sigma_2$ and $\sigma_3$ are the standard deviations of $w^2$ and $|w|^3$, respectively. 
\end{thm}

\subsection{Proof of Theorem~\ref{thm:lb_random}}\label{subsec:balanced}
As argued in Section~\ref{subsec:unbal}, it suffices to establish
the lower bound on influence when $w_0 = 0$. To establish the theorem
we argue similarly as in the proof of Lemma~\ref{lem:mainlb}. 
We consider the case when the distribution $\calD$ over $\reals$ is symmetric around $0$ and with variance $1$.
That is, $\E_{w \sim \calD}[w^2] = 1$.
Let $\mu_3 := \E_{w \sim \calD}[|w|^3]$.
Further, $\sigma_2$ and $\sigma_3$ denote the standard deviation
of $w^2$ and $|w|^3$, respectively. 
That is, $\sigma_2^2 = \E[(w^2-\E[w^2])^2]$
and $\sigma_3^2 = \E[(|w|^3 - \E[|w|^3])^2 ]$. 

The following lemma gives a lower bound on the $i$-th influence. 
\begin{lem}
  \label{lem:bound-rademacher-sums-random}
  Let $w_j \sim \calD$ for $1 \leq j \leq n$. For all $i \in [n]$,
  $\alpha\in \reals^+$, and $\delta > 0$,  with probability at 
  least $1 - e^{-\delta^2 n \mu_3^2/\sigma_3} - 2e^{-\delta^2 n/\sigma_2}$
  over the choices of $w_j$'s,
  \[ \Pr_{x \in \sbool^n}\left[\left|\sum_{1\leq j \leq n \colon j \neq i} w_jx_j\right| \leq  \alpha\right] \geq \frac{\theta}{\sqrt{n}} - O\left(\frac{1}{n^{2/3}}\right),\]
  where $\theta = (\alpha - \mu_3(1+ \frac{2\delta}{1-\delta}))/\sqrt{2\pi (1+\delta)}$.
\end{lem}
\begin{proof}
To start we have $w_1, \dots, w_n$ that are independently and identically distributed according to $\mathcal{D}$. Define random variables $A_i = (\sum_{j\neq i} |w_j|^3)/(n-1)$ and $B_i = (\sum_{j\neq i} w_j^2)/(n-1)$. 
Thus $\E[A_i] = \mu_3$ and $\E[B_i] = 1$.  
Using Bernstein's inequality (see Corollary~2.11 in Chapter 2 of \cite{BLM13}) 
on independent symmetric random variables $w_j$'s
  we obtain
\begin{align}\label{eq:bound-on-squares}
   (1-\delta)(n-1) \leq \sum_{1 \leq j \leq n\colon j\neq i} w_j^2 \leq (1+\delta)(n-1) \enspace ,
\end{align}
  with probability at least $1 - 2e^{-\delta^2 n/\sigma_2}$, where $\sigma_2^2 = \mathrm{Var}_{w\sim \calD}(w^2)$. 
  Similarly, we also have
  \begin{align}
    \label{eq:bound-on-cubes}
    \sum_{1 \leq j \leq n \colon j \neq i} w_j^3 \leq (1+\delta)C(n-1) \enspace ,
  \end{align}
  with probability at least  $1 - e^{-2\delta^2 n \mu_3^2/\sigma_3}$, where $\sigma_3^2 = \mathrm{Var}_{w\sim \calD}(|w|^3)$. Thus, from the union bound,
  it follows that both equations,~\eqref{eq:bound-on-squares} and
  \eqref{eq:bound-on-cubes}, holds with probability at least $1 - e^{-\delta^2 n \mu_3^2/\sigma_3} - 2e^{-\delta^2 n/\sigma_2}$, and thus with this probabilty 
we have $(A_i/B_i) \leq \left(1 + \frac{2\delta}{1-\delta}\right)\mu_3.$ 
Recall from Lemma~\ref{lem:mainlb} we know 
\[\Pr_{x\in \sbool^n}\left[\left|\sum_{j\leq n: j\neq i}w_jx_j\right| \leq \alpha\right] \geq \frac{1}{\sqrt{2\pi (n-1)B_i}}\left(\alpha - \frac{A_i}{B_i}\right) - O\left(\frac{1}{n^{2/3}}\right) \enspace . \]
Thus plugging in the bound on $(A_i/B_i)$ in the above inequality we obtain
the lemma, 
\[\Pr_{x\in \sbool^n}\left[\left|\sum_{j\leq n: j\neq i}w_jx_j\right| \leq \alpha\right] \geq \frac{\alpha - (1+\frac{2\delta}{1-\delta})\mu_3}{\sqrt{2\pi (1+\delta)(n-1)}} - O\left(\frac{1}{n^{2/3}}\right) \enspace . \]
\qed\end{proof}

We now proceed to complete the proof of the theorem.
For any $\alpha \in \reals^+$,
\[\inf(f) = \sum_i \inf_i(f) \geq \sum_{i: w_i \geq \alpha} \inf_i(f) \geq \sum_{i: w_i \geq \alpha} \Pr_{x \in \sbool^n}\left[\left|\sum_{j \leq n \colon j \neq i} w_jx_j\right| \leq  \alpha\right] \enspace .\]
Thus, from Lemma~\ref{lem:bound-rademacher-sums-random}, we have with probabiltity at least $1 - e^{-\delta^2 n \mu_3^2/\sigma_3} - 2e^{-\delta^2 n/\sigma_2}$ over the choices of $w_j$'s,
\[\inf(f)  \geq N_{\alpha}\cdot \left(\frac{\alpha - \mu_3\left(1+ \frac{2\delta}{1-\delta}\right)}{\sqrt{2(n-1)\pi (1+\delta)}} - O\left(\frac{1}{n^{2/3}}\right)\right) \enspace, \]
where $N_{\alpha}$ is the number of $w_i$'s that are bigger than $\alpha$. 
Now the following lemma shows that with probability at least $1-2e^{-\frac{n}{4} p_{\calD,1}(\alpha)^2}$ we have $N_{\alpha} \geq \frac{n}{2}\cdot p_{\mathcal{D},1}(\alpha)$, and that establishes the theorem. 
\begin{lem}
  \label{lem:large-wt}
  Fix an $\alpha \in \reals^{+}$. 
  If we sample $n$ points according to $\calD$
  then with probability at least $1-2e^{-\frac{n}{4} p_{\calD,1}(\alpha)^2}$
  the number of points in the interval $[\alpha,\infty)$ is greater than
    $\frac{n\cdot p_{\calD,1}(\alpha)}{2}$ and
    less than $\frac{3n\cdot p_{\calD,1}(\alpha)}{2}$.
\end{lem}
\begin{proof}
  Let $\mathbf{1}_i$ be the indicator function of the event that
  the $i$-th sample
  lies in the interval $[\alpha,\infty)$.
    Define $X = \sum_{j=1}^n \mathbf{1}_j$.
    Note that for each $j$, $\E[I_j] = p_{\calD,1}(\alpha)$,
    so using Chernoff bound on $X$ we get, 
    \[ \Pr\left[\frac{n\cdot p_{\calD,1}(\alpha)}{2} \leq X \leq \frac{3n\cdot p_{\calD,1}(\alpha)}{2}\right] \geq 1 - 2e^{-\frac{1}{4}n( p_{\calD,1}(\alpha))^2}.\]
    \qed\end{proof}
    
\subsection{Fourier Entropy Conjecture for a Random Linear Threshold Function}
\label{subsec:lb-normal}
In this section, we give two natural examples of distributions on the $w_i$'s
under which the FEI conjecture holds with high probability.
First we consider the  uniform distribution on the closed interval $[-1,1]$, and
then the normal distribution $N(0,1)$. Together this completes the proof of Theorem~\ref{thm:main}.

\subsubsection{Uniform distribution on $[-1,1]$: $\mathcal{U}(-1,1)$}
\begin{coro}\label{cor:FEIUniform} 
If $w_i \sim \mathcal{U}(-1,1)$ for $i \in \{0, \ldots, n\}$, then the $FEI$ holds with high probability. 
\end{coro}
\begin{proof}
  Since $\E_{x\sim \mathcal{U}(-1,1)}[x^2] = 1/3$, when $w_0, \dots, w_n$ are drawn
  independently from $\mathcal{U}[-1,1]$, by Chernoff bound we have 
  with high probability 
  $\sum_{i=0}^n\E[w_i^2] = \Omega(n)$. Therefore, from Corollary~\ref{cor:random-Khint}, we have with high probability $\inf(f) = \Omega(\sqrt{n})$.
  Using the upper bound on the 
  Fourier entropy from Theorem~\ref{thm:entub} we have our result. 
\qed\end{proof}

\subsubsection{Normal distribution : $N(0,1)$}
We note that for the normal distribution,
unlike the case of uniform distribution,
the FEI conjecture does not follow from Khintchine's inequality.
Indeed, Corollary~\ref{cor:random-Khint} does not give us what we want,
i.e., a high probability on the event $\inf(f) = \Omega(\sqrt{n})$.
On the other hand, if we try to boost the probability
we end up with a weaker lower bound on the influence.
To see this, observe that when $w_1, \dots, w_n$ are drawn independently
from $N(0,1)$, 
\begin{enumerate}
\item $\E[\sum_i w_i^2] = n$ (since $\E[X^2] = 1$ for $X \sim N(0,1)$).
  This implies that with high probability $ \sum w_i^2 = \Omega(n)$,
  and hence $\sqrt{\sum w_i^2}= \Omega(\sqrt{n})$.
\item The probability that $\max_{i=1}^n |w_i| = O(\sqrt{\log n})$ is $1-o(1)$.
  Using the fact $\Pr[X > x] \leq  e^{-x^2/2}$ when $X \sim N(0,1)$,
  we have $\Pr[w_i > \sqrt{2c\log n}] \leq \frac{1}{n^c}$.
  Since $w_i$'s are drawn independently, the probability that
  all the $w_i$'s are less than $\sqrt{2c\log n}$ is lower bounded
  by $1-\frac{1}{n^{c-1}}$, where $c > 1$ is a constant. 
\end{enumerate}
Applying Corollary~\ref{cor:random-Khint} to this now gives us $\mathsf{Inf}(f) \geq \Omega\left(\sqrt{\frac{n}{\log n}}\right)$ with $1-o(1)$ probability.
Moreover, the max-central limit theorem 
implies that for large enough $n$,
with high probability, $\max_i|w_i| = \Theta(\sqrt{\log n})$ (see Example~10.5.3 in \cite{DN03}). 
However using the other technique, namely Theorem~\ref{thm:lb_random}, we can obtain our desired result. 
\begin{coro}\label{cor:FEINormal} 
If $w_i \sim N(0,1)$ for $i \in \{0, \ldots, n\}$, then the $FEI$ holds with probability at least $1 - e^{-\Omega(n)}$. 
\end{coro}
\begin{proof}
  We use the following well known bounds on moments of $N(0,1)$.
  \[
  \E_{w \sim N(0,1)}[|w|^r] = \begin{cases}
    \sqrt{\frac{2}{\pi}}\left(r-1\right)!! & \text{if $r$ is odd}, \\
    (r-1)!! & \text{if $r$ is even.}
  \end{cases}  
  \]
  Therefore, we have $\mu_3 = \frac{2\sqrt{2}}{\sqrt{\pi}}$, $\sigma_2^2 = \E[w^4] - (\E[w^2])^2 = 2$, and $\sigma_3^2 = \E[|w|^6] - (\E[|w|^3])^2 = (15-\frac{8}{\pi})$. Further, we set $\alpha = \mu_3(2+\frac{2\delta}{1-\delta})$, for any $\delta > 0$. Hence,
  \[p_{\calD,1}(\alpha) \geq 2\cdot \frac{1}{\sqrt{2\pi}} e^{-\frac{\alpha^2}{2}} \left(\frac{1}{\alpha}-\frac{1}{\alpha^3} \right) = \Omega(1) \enspace .\]
  The first inequality is easily established; for example, see Exercise~1, Chapter~7 in \cite{Feller68}. Thus, from Theorem~\ref{thm:lb_random} we have that $\inf(f) = \Omega(\sqrt{n})$ and now using Theorem~\ref{thm:entub} we have our result. 
\qed\end{proof}
\begin{remark}
  \label{rem:generic-distribution}
  We note that the proof of Theorem~\ref{thm:lb_random} shows that
  the FEI conejcture holds with high probability,
  as long as $\calD$ (with $\mu=0$ and $\sigma^2=1$)
  statisfies the following properties:
  $(i)$~$\E_{w\sim \calD}[|w|^3]$ is finite,
  and $(ii)$~$\Pr_{w\sim \calD}[|w| > \E_{w\sim \calD}[|w|^3]] = \Omega(1)$.
  In particular, our proof holds for \emph{truncated} distributions,
  and the optimality of the constant in Shevtsova's theorem
  (Theorem~\ref{thm:shevtsova})
  is crucial in such cases to establish property~$(ii)$.  
\end{remark}

%
%
%

%% file: conclusion.tex
\section{Conclusion and Open Problems}
\label{sec:conclusion}


We proved in this paper that for a random linear threshold function $f(x):=\sign(w_0 + \sum_i w_ix_i)$, where the $w_i$'s are drawn from a distribution
that has some `nice' properties (Remark~\ref{rem:generic-distribution}),  the
FEI conjecture holds with high probability. Moreover, we show that the uniform distribution over an interval, say $[-1,1]$, and the normal distribution $N(0,1)$ satisfy these 'nice' properties. Indeed, we established that a random LTF sampled
according to these distributions has influence $\Omega(\sqrt{n})$.  When combined with the $O(\sqrt{n})$ upper bound 
(Theorem~\ref{thm:entub}) from previous work on the Fourier entropy of \emph{all} LTF's, we conclude that the FEI conjecture holds for the random LTF's sampled as above.
In the process, we obtain non-trivial lower bounds on the influence of $f$
in terms of the $w_i$'s.

An obvious open question is to prove that the FEI conjecture holds for all LTF's. While our current techniques seem far from sufficient to achieve this goal, a natural question is to prove a generalization of Theorem~\ref{thm:entub}, and give an upper bound on the
Fourier entropy in terms of $w_i$'s. We believe it is possible to obtain such a general upper bound as a function of a suitable notion of ``skewness'' of the weights $w_i$. 

Another natural question is whether one can show that FEI conjecture holds for a random \emph{polynomial} threshold function of degree $d$. A polynomial threshold function of degree $d$ is defined as the sign of a degree-$d$  polynomial. For a degree-$d$ PTF
$f(x) = \sign(p(x_1,\ldots ,x_n))$, observe that 
$\inf_i(f)=\Pr_x[\sign(p(x)) \neq \sign(p(x^i))]=\Pr_x[|p(x)|< 2|D_ip(x)|]$,
where $D_ip(x)$ is the partial derivative of $p(x)$ with respect to $x_i$.
Hence, to lower bound the influence, a standard approach would be to show that $p(x)$ has
certain \emph{concentration} properties and $D_ip(x)$ has certain \emph{anti-concentration} properties, for a random $p$. However, we currently do not know such strong enough concentration/anti-concentration bounds. 
Our techniques do not seem to generalize even to polynomial threshold functions of degree $2$. Such generalizations seem to require more powerful tools. We also note that currently the best bound on the Fourier entropy of a degree-$d$ PTF is $O(\sqrt{n} (\log n)^{O(d\log d)})$ where the constant in $O(\cdot)$ depends on the degree $d$. (It follows from Kane's bound on the average sensitivity of a PTF~\cite{Kane14}.)

